\documentclass[12pt]{article}
\usepackage{graphicx}
\usepackage{xcolor}
\usepackage{enumerate}
\usepackage{amsmath}
\usepackage{amssymb}
\usepackage{amsthm}
\usepackage{geometry}
\usepackage{hyperref}

\newtheorem{definition}{Definition}
\newtheorem{theorem}{Theorem}
\newtheorem{proposition}{Proposition}

\title{
A Geometric Family of Correlations Containing the Quantum Singlet}

\author{E. Aldo Arroyo\thanks{aldo.arroyo@ufabc.edu.br},\\
    Centro de Ci\^{e}ncias Naturais e Humanas, Universidade Federal do ABC,\\
    Santo Andr\'{e}, 09210-170 S\~{a}o Paulo, SP, Brazil.}

\date{\today}

\begin{document}

\maketitle

\begin{abstract}

We introduce a geometrically constrained hidden-variable framework
that generates a family of correlations parametrized by a boundary
function, within which the quantum singlet correlation appears as
a particular member. Exact expressions for the correlation
function are derived. Several structural results are established,
including admissibility conditions, symmetry properties, a
universal stationary point of the associated CHSH function, and an
exact relation between the CHSH value at $\nu=\pi/4$ and a
geometric contrast measure defined on the underlying
hidden-variable distributions. Rather than treating the quantum
singlet correlation as an isolated target to be reproduced, the
present framework places it within a broader geometric structure
of correlations. These results suggest the existence of a
nontrivial geometric structure underlying the family of
correlations and motivate the search for a principle capable of
selecting the quantum singlet solution from within that family.

\end{abstract}

\section{Introduction}
\label{sec:intro}

The violation of Bell inequalities \cite{bell1964,chsh1969} by
quantum mechanical correlations is often interpreted as
demonstrating that any underlying local realistic theory must be
incompatible with quantum mechanics. This conclusion, however,
rests on several assumptions, one of the most critical being that
the hidden variables $\lambda$ are statistically independent of
the measurement settings $\mathbf{a},\mathbf{b}$, i.e.
\[
\rho(\lambda|\mathbf{a},\mathbf{b}) = \rho(\lambda).
\label{eq:mi_assumption}
\]
This condition is commonly referred to as measurement
independence, statistical independence, or the free-choice
assumption \cite{hall2010,hall2011,hall2015}.

The violation of measurement independence has often been regarded
as problematic, since it is frequently interpreted as requiring a
conspiracy-like correlation between hidden variables and
measurement settings \cite{bell1981,mermin1985}.

This concern is often motivated by the view that correlations
between hidden variables and measurement settings necessarily
require some form of causal connection between them. As Kupczynski
\cite{kupczynski2023} and others
\cite{brans1988,palmer2024,friedman2019,arroyo2025} have
emphasized, however, correlation alone does not imply causation;
such correlations may arise from a common cause or, as we will
argue, from a global geometric constraint on the space of
physically realizable configurations.

The logical gap in the conspiracy objection was already noted by
Brans \cite{brans1988}, who exhibited an explicit local
deterministic model of singlet correlations by relaxing
measurement independence. Hall \cite{hall2015} further quantified
that less than $1/15$ bits of prior correlation suffice for a
local model of the singlet state and argued that measurement
independence should be distinguished from operational notions of
free choice.

More recently, Hance, Hossenfelder and Palmer \cite{hance2022}
introduced the concept of a supermeasured theory, in which
Bell-statistical independence is violated not because of a causal
correlation between hidden variables and settings, but because the
measure on the state space is nontrivial. In such theories, the
set of physically admissible tuples $(\lambda,\mathbf a,\mathbf
b)$ is a proper subset of the Cartesian product
$\Lambda\times\mathcal A\times\mathcal B$, reflecting a global
compatibility condition rather than a causal conspiracy.

Related ideas have also been explored in other approaches to the
foundations of quantum theory. Retrocausal models allow hidden
variables to depend on both past and future boundary conditions,
thereby relaxing the assumptions entering Bell's theorem without
invoking superluminal influences \cite{price2012,adlam2022}.
Closely related all-at-once or Lagrangian-based frameworks
emphasize global spacetime constraints rather than dynamical
causal mechanisms \cite{wharton2010,wharton2015}. In a different
direction, Palmer's invariant-set approach attributes the apparent
violation of Bell-statistical independence to geometric
restrictions on a fractal subset of state space
\cite{palmer2009,palmer2025}. Although these approaches differ
substantially in their physical interpretation, they share the
idea that the space of physically realizable configurations may be
more restricted than is usually assumed.

In this work, we propose a geometric framework in which, rather
than postulating an explicit dependence of the hidden-variable
distribution on the measurement settings, we start from a
geometrically restricted admissible region
\[
\Gamma \subset \Lambda\times\mathcal{A}\times\mathcal{B},
\]
defined explicitly in terms of a boundary function $\chi(\gamma)$,
where $\gamma$ is the angular separation between the measurement
directions. The construction is fully local and deterministic: the
outcomes are given by
$A(\mathbf{a},\lambda)=\operatorname{sign}(\mathbf{a}\cdot\lambda)$
and
$B(\mathbf{b},\lambda)=-\operatorname{sign}(\mathbf{b}\cdot\lambda)$.
Admissible hidden variables are those whose azimuthal angle $\phi$
falls outside two antipodal forbidden intervals whose sizes are
controlled by $\chi(\gamma)$.

Let us emphasize that this geometric restriction does not imply
superluminal signaling or any dynamical non-local mechanism. The
local response functions $A(\mathbf a,\lambda)$ and $B(\mathbf
b,\lambda)$ depend only on their respective local measurement
settings and the hidden variable. Instead, the present framework
is conceptually closer to approaches based on global constraints,
variational principles, or spacetime-wide consistency conditions,
in which the physically realizable configuration space is
restricted from the outset rather than arising from a purely local
dynamical evolution. From this perspective, Bell-statistical
independence is relaxed through a restriction on the admissible
configuration space rather than through an explicit causal
connection between settings and hidden variables.

Crucially, the admissible region $\Gamma_\gamma$ for fixed
settings retains the antipodal symmetry
\[
\Phi_\gamma+\pi = \Phi_\gamma,
\]
which guarantees unbiased marginals $\langle A\rangle=\langle
B\rangle=0$ for any admissible $\chi$. Thus, while the
hidden-variable distribution depends on \(\gamma\), the local
outcomes remain perfectly random and no local statistical bias is
introduced by the geometric restriction.

We derive the exact correlation function $E(\gamma)$ for any
admissible $\chi$,
\[
E(\gamma) =
\begin{cases}
\displaystyle \frac{4\gamma-\pi-2\chi(\gamma)}{3\pi-2\chi(\gamma)}, & 0<\gamma<\frac{\pi}{2},\\[8mm]
\displaystyle \frac{4\gamma-\pi-2\chi(\gamma)}{\pi+2\chi(\gamma)},
& \frac{\pi}{2}<\gamma<\pi.
\end{cases}
\]

The quantum singlet correlation $E(\gamma)=-\cos\gamma$ emerges as
a particular member of the admissible family, corresponding to a
specific boundary function \(\chi_{\mathrm s}(\gamma)\). This
observation shifts the problem from reproducing the quantum
correlation to understanding what geometric principle, if any,
singles out the quantum singlet solution from the broader
admissible family.

Within the reflection-symmetric subclass
$\mathcal{R}_{\mathrm{adm}}$ (defined by
$\chi(\pi-\gamma)=\pi-\chi(\gamma)$), we prove that the CHSH
function $S(\nu)=3E(\nu)-E(3\nu)$ satisfies $S'(\pi/4)=0$ for
every admissible boundary function
$\chi\in\mathcal{R}_{\mathrm{adm}}$, establishing $\nu=\pi/4$ as a
universal stationary point of the reflection-symmetric admissible
family. Moreover, we establish an exact relation between the CHSH
value at the distinguished point $\nu=\pi/4$ and the geometric
contrast $D(\pi/4)$ between the hidden-variable ensembles
associated with complementary angles:
\[
\left|S\!\left(\frac{\pi}{4}\right)\right| = 2 +
3\,D\!\left(\frac{\pi}{4}\right).
\]
This relation identifies a direct correspondence between the
strength of CHSH correlations and the distinguishability of
complementary hidden-variable ensembles.

The paper is organized as follows. In Section~\ref{sec:geometry}
we define the hidden-variable space, the local observables, the
admissible boundary function $\chi(\gamma)$, and the associated
admissible region, area, and distribution.
Section~\ref{sec:correlations} presents the general correlation
function and proves the vanishing marginals. In
Section~\ref{sec:symmetric_family} we introduce the
reflection-symmetric admissible family and examine three important
members: the constant model, the quadratic model, and the quantum
singlet solution. Section~\ref{sec:chsh_structure} analyzes the
CHSH function, establishes the universal stationary point, and
derives the geometric interpretation of the CHSH value at
$\nu=\pi/4$. Finally, Section~\ref{sec:conclusions} discusses the
implications of our findings and outlines open questions,
including the search for a geometric principle that might select
the quantum singlet solution.

\section{Geometric Construction of the Hidden-Variable Model}
\label{sec:geometry}

Let the hidden variable be a unit vector in $\mathbb{R}^{3}$,
\begin{equation}
\boldsymbol{\lambda} = (\sin\theta\cos\phi,\;
\sin\theta\sin\phi,\; \cos\theta) \in S^{2} \equiv \Lambda,
\label{eq:lambda_spherical}
\end{equation}
where $\theta\in[0,\pi]$ and $\phi\in[0,2\pi)$ are the standard
spherical coordinates.

Let the unit vectors $\mathbf a,\mathbf b$ denote the measurement
directions of Alice and Bob. Without loss of generality, we choose
these vectors as follows
\begin{equation}
\mathbf a=(1,0,0), \qquad \mathbf b=(\cos\gamma,\sin\gamma,0),
\label{eq:a_b_fixed}
\end{equation}
where
\begin{equation}
\gamma = \arccos(\mathbf a\cdot\mathbf b) \in[0,\pi]
\label{eq:gamma_def}
\end{equation}
is the angular separation between the measurement settings.

The entire construction developed below depends only on the
relative orientation of the measurement settings. Consequently,
all geometric quantities may be regarded as functions of the
single angular variable $\gamma$.

Throughout this work, the symbol $\phi$ denotes the azimuthal
coordinate of the hidden variable $\boldsymbol{\lambda}$, whereas
$\gamma$ denotes the angular separation between the measurement
settings.

The hidden-variable model is completed by specifying the local
response functions associated with the measurement directions
$\mathbf a$ and $\mathbf b$.

For the choice of coordinates introduced above, the measurement
outcomes are taken to be deterministic functions of the hidden
variable $\boldsymbol{\lambda}$ and are defined by

\begin{equation}
A(\mathbf a,\boldsymbol{\lambda}) = \operatorname{sign} \!\left(
\mathbf a\cdot\boldsymbol{\lambda} \right) =
\operatorname{sign}(\cos\phi), \label{eq:A_definition}
\end{equation}

and

\begin{equation}
B(\mathbf b,\boldsymbol{\lambda}) = -\operatorname{sign} \!\left(
\mathbf b\cdot\boldsymbol{\lambda} \right) = -\operatorname{sign}
\!\left[ \cos(\phi-\gamma) \right]. \label{eq:B_definition}
\end{equation}

The minus sign in the definition of $B$ is introduced so that the
quantum singlet correlation corresponds to an anti-correlated
configuration when the measurement directions coincide.

By construction,

\begin{equation}
A(\mathbf a,\boldsymbol{\lambda}), \; B(\mathbf
b,\boldsymbol{\lambda}) \in \{-1,+1\},
\label{eq:binary_observables}
\end{equation}

for almost every hidden variable $\boldsymbol{\lambda}\in\Lambda$.

The corresponding product observable is therefore

\begin{equation}
A(\mathbf a,\boldsymbol{\lambda}) B(\mathbf
b,\boldsymbol{\lambda}) = - \operatorname{sign} \!\left[
\cos\phi\, \cos(\phi-\gamma) \right]. \label{eq:AB_definition}
\end{equation}

The expectation values and correlation functions derived in the
following sections are obtained by averaging the observables
(\ref{eq:A_definition})--(\ref{eq:AB_definition}) with respect to
the admissible probability distribution
$\rho_{\gamma}(\boldsymbol{\lambda})$ defined in
Definition~\ref{def:admissible_distribution}.

\begin{definition}[Admissible Boundary Function]
\label{def:admissible_chi}

A function $\chi(\gamma)$ is called admissible if it satisfies

\begin{equation}
\frac{\pi}{2} < \chi(\gamma) < \frac{\pi}{2}+\gamma, \qquad
0<\gamma<\frac{\pi}{2}, \label{eq:admissible_low}
\end{equation}

and

\begin{equation}
\gamma-\frac{\pi}{2} < \chi(\gamma) < \frac{\pi}{2}, \qquad
\frac{\pi}{2}<\gamma<\pi. \label{eq:admissible_high}
\end{equation}

The function $\chi(\gamma)$ specifies the boundary between
admissible and forbidden azimuthal sectors.
\end{definition}

Different admissible choices of $\chi$ generate different
admissible hidden-variable models and consequently different
correlation functions.

A particular admissible choice of $\chi$ will later be shown to
reproduce the quantum singlet correlation.

The two admissible branches meet at $\gamma=\frac{\pi}{2}$.
Therefore, any continuous admissible boundary function satisfies
$\chi\!\left(\frac{\pi}{2}\right) = \frac{\pi}{2}$. Moreover, the
admissibility inequalities imply
\[
\lim_{\gamma\to0^{+}} \chi(\gamma) = \lim_{\gamma\to\pi^{-}}
\chi(\gamma) = \frac{\pi}{2}.
\]

Hence, any continuous extension of $\chi$ to the closed interval
$[0,\pi]$ satisfies
\begin{equation}
\chi(0) = \chi\!\left(\frac{\pi}{2}\right) = \chi(\pi) =
\frac{\pi}{2}. \label{eqvalueschis} \end{equation}

Thus, every continuous admissible boundary function passes through
the three distinguished points
\[
(0,\pi/2), \qquad (\pi/2,\pi/2), \qquad (\pi,\pi/2).
\]

The differences between admissible models therefore arise entirely
from the behavior of $\chi$ in the interior of the interval
$(0,\pi)$.

\begin{figure}
\centering
\includegraphics[width=4.8in,height=115mm]{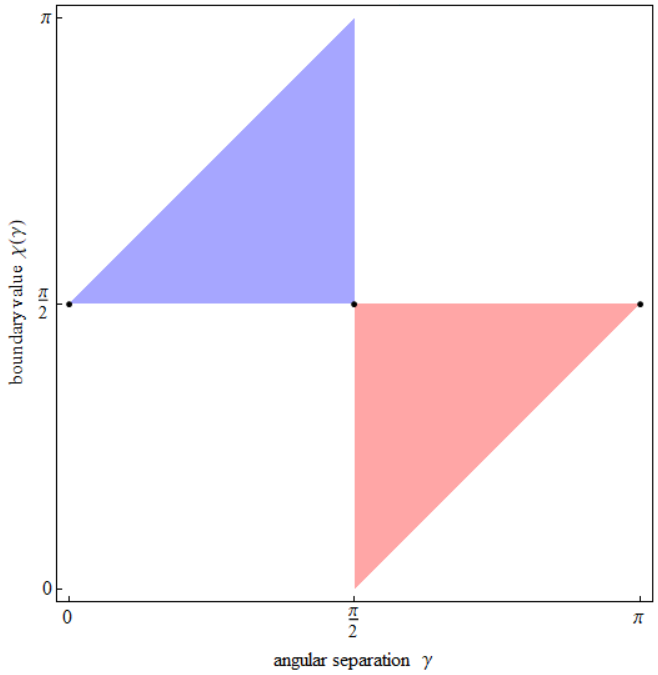}
\caption{Space of admissible boundary functions. The shaded region
defines the allowed values of $\chi(\gamma)$ for each angular
separation $\gamma$. Continuous admissible boundary functions form
a family of curves contained within this domain, all sharing the
fixed points $(0,\pi/2)$, $(\pi/2,\pi/2)$ and $(\pi,\pi/2)$. }
\label{fig:admissible_chi}
\end{figure}

Figure~\ref{fig:admissible_chi} provides a geometric
representation of the admissible boundary functions. The shaded
region defines the allowed domain in the $(\gamma,\chi)$ plane,
while continuous admissible functions correspond to curves
entirely contained within this domain.

Each admissible function determines a particular hidden-variable
model through the exclusion of two opposite azimuthal sectors from
the sphere. This construction preserves antipodal symmetry and
generates a family of admissible correlation models parametrized
by $\chi(\gamma)$.

\begin{definition}[Admissible Region]
\label{def:Gamma}

Let $\chi$ be an admissible boundary function.

For a given angular separation $\gamma$, define

\begin{align}
\Gamma_{\gamma} =
\begin{cases}
\Bigl\{ \boldsymbol{\lambda}\in\Lambda \;\Big|\; \phi \notin
\bigl(\frac{\pi}{2},\,\chi(\gamma)\bigr) \cup
\bigl(\frac{3\pi}{2},\,\chi(\gamma)+\pi\bigr) \Bigr\}, & 0<\gamma<\frac{\pi}{2},\\[6mm]
\Bigl\{ \boldsymbol{\lambda}\in\Lambda \;\Big|\; \phi \notin
\bigl(\chi(\gamma),\,\frac{\pi}{2}\bigr) \cup
\bigl(\chi(\gamma)+\pi,\,\frac{3\pi}{2}\bigr) \Bigr\}, &
\frac{\pi}{2}<\gamma<\pi.
\end{cases} \label{eq:Gamma_def}
\end{align}

The set

\[
\Gamma_{\gamma} \subset S^{2}
\]

will be called the admissible hidden-variable region associated
with the angular separation $\gamma$.
\end{definition}

Geometrically, $\Gamma_{\gamma}$ is obtained by removing two
azimuthal intervals from the hidden-variable sphere. The size and
location of these forbidden sectors are controlled by the boundary
function $\chi(\gamma)$.

As $\gamma$ varies, the admissible region changes continuously,
generating a family of geometrically distinct hidden-variable
ensembles.

\begin{definition}[Admissibility Function]
\label{def:admissibility_function}

For each angular separation $\gamma$, define

\begin{equation}
\mathcal I_{\gamma}(\boldsymbol{\lambda}) =
\begin{cases}
1, & \boldsymbol{\lambda}\in\Gamma_{\gamma},
\\[2mm]
0, & \boldsymbol{\lambda}\notin\Gamma_{\gamma}.
\end{cases}
\label{eq:admissibility_function} \end{equation}

\end{definition}

For practical calculations it is convenient to express the
admissibility function in terms of the azimuthal coordinate
$\phi$.

Introducing the Heaviside step function $H(x)$, the indicator
function of the admissible region may be written explicitly as

\begin{align}
\mathcal I_\gamma(\phi) = 1 - \operatorname{sign}\!\left(
\frac{\pi}{2} - \chi(\gamma) \right) \Bigg[
H\!\bigl(\phi-\chi(\gamma)\bigr) +
H\!\bigl(\phi-\chi(\gamma)-\pi\bigr) -
H\!\left(\phi-\frac{\pi}{2}\right) -
H\!\left(\phi-\frac{3\pi}{2}\right) \Bigg].
\label{eq:admissibility_representation}
\end{align}

This representation will be particularly useful for the evaluation
of the admissible area, probability distributions and correlation
integrals in spherical coordinates.

\begin{definition}[Admissible Area]
\label{def:admissible_area}

The area of the admissible region is defined by

\begin{equation}
\Omega(\gamma) = \int_{\Lambda} \mathcal
I_{\gamma}(\boldsymbol{\lambda}) \,d\lambda .
\label{eq:admissible_area} \end{equation}

\end{definition}

The quantity $\Omega(\gamma)$ measures the effective volume of
hidden-variable states that remain accessible after imposing the
geometric restrictions determined by $\chi(\gamma)$.

\begin{definition}[Admissible Distribution]
\label{def:admissible_distribution}

The normalized hidden-variable distribution associated with
$\gamma$ is

\begin{equation}
\rho_{\gamma}(\boldsymbol{\lambda}) = \frac{ \mathcal
I_{\gamma}(\boldsymbol{\lambda}) }{ \Omega(\gamma) }.
\label{eq:admissible_distribution} \end{equation}

\end{definition}

By construction,

\[
\int_{\Lambda} \rho_{\gamma}(\boldsymbol{\lambda}) \,d\lambda = 1.
\]

Therefore, each angular separation $\gamma$ induces a probability
distribution on the hidden-variable space. The dependence of
$\rho_{\gamma}$ on $\gamma$ encodes the geometric restrictions
imposed by the admissible boundary function $\chi$.

\section{Marginal Distributions and Correlation Functions}
\label{sec:correlations}

Having introduced the admissible hidden-variable distributions, we
now derive the statistical predictions associated with an
arbitrary admissible boundary function $\chi(\gamma)$.

Our first objective is to determine the local expectation values

\[
\langle A\rangle, \qquad \langle B\rangle,
\]

and the correlation function

\[
E(\gamma) = \langle AB\rangle .
\]

A central feature of the construction is that the antipodal
symmetry of the admissible regions leads to unbiased marginal
distributions independently of the choice of admissible boundary
function.

The following results hold for every admissible model.

\begin{theorem}[Vanishing Marginals]
\label{thm:vanishing_marginals}

For every admissible boundary function $\chi(\gamma)$ satisfying
Definition~\ref{def:admissible_chi}, the local expectation values
satisfy

\begin{equation}
\label{eq:vanishing_marginals} \langle A\rangle = \langle B\rangle
= 0.
\end{equation}

\end{theorem}

\begin{theorem}[General Correlation Function]
\label{thm:general_correlation}

For every admissible boundary function $\chi(\gamma)$ satisfying
Definition~\ref{def:admissible_chi}, the correlation function

\[
E(\gamma) = \langle AB\rangle
\]

admits the exact representation

\begin{equation}
\label{eq:general_correlation} E(\gamma) =
\begin{cases}
\displaystyle \frac{4\gamma - \pi - 2\chi(\gamma)}{3\pi - 2\chi(\gamma)}, & 0<\gamma<\frac{\pi}{2},\\[8mm]
\displaystyle \frac{4\gamma - \pi - 2\chi(\gamma)}{\pi +
2\chi(\gamma)}, & \frac{\pi}{2}<\gamma<\pi,
\end{cases}
\end{equation}

where the right-hand side depends only on the admissible boundary
function $\chi(\gamma)$.

\end{theorem}

The proof of Theorem~\ref{thm:vanishing_marginals} relies only on
the antipodal symmetry of the admissible regions, whereas
Theorem~\ref{thm:general_correlation} requires the explicit
evaluation of the corresponding hidden-variable integrals.

The detailed proofs of these theorems are presented in
Appendices~\ref{subsec:antipodal} and
\ref{subsec:proof_general_correlation}.

Theorem~\ref{thm:vanishing_marginals} shows that the geometric
restrictions introduced in Section~\ref{sec:geometry} do not
generate any local statistical bias. From the perspective of
either observer, the measurement outcomes remain completely
random, with the values $+1$ and $-1$ occurring with equal
probability.

Consequently, Alice and Bob cannot infer the presence of the
underlying geometric restrictions from their local data alone. Any
information about the admissible boundary function $\chi(\gamma)$
is hidden at the level of the marginal distributions.

Theorem~\ref{thm:general_correlation} reveals where the geometric
structure becomes observable. Although all admissible models share
identical local statistics, they generally differ in their
correlation function $E(\gamma)$. The entire family of admissible
correlations is therefore encoded in the boundary function
$\chi(\gamma)$.

In this sense, the admissible boundary function plays the role of
a geometric parameter that determines how the hidden-variable
restrictions manifest themselves in the observable correlations
while leaving the local marginals unchanged.

The remainder of this work is devoted to the study of the
geometric properties of this family and to the identification of
particular boundary functions that reproduce physically
distinguished correlations, including the quantum singlet case.

\section{Reflection-Symmetric Admissible Family}
\label{sec:symmetric_family}

The general admissible family introduced in
Section~\ref{sec:geometry} contains a large variety of correlation
models.

However, several physically relevant examples, including the
quantum singlet solution and the quadratic model introduced later,
possess an additional reflection symmetry.

This observation motivates the introduction of a smaller function
space on which stronger structural results can be established.

\begin{definition}[Reflection-Symmetric Admissible Family]
\label{def:symmetric_family}

Let

\[
\mathcal{R}_{\mathrm{adm}}
\]

denote the set of continuously differentiable admissible boundary
functions satisfying the reflection symmetry

\begin{equation}
\chi(\pi-\gamma) = \pi-\chi(\gamma), \qquad 0\le\gamma\le\pi.
\label{eq:reflection_symmetry}
\end{equation}

Explicitly,

\[
\mathcal{R}_{\mathrm{adm}} = \left\{ \chi\in C^{1}([0,\pi]) \;
\Big| \; \chi \text{ is admissible and satisfies }
\eqref{eq:reflection_symmetry} \right\}.
\]

\end{definition}

Eq.~\eqref{eq:reflection_symmetry} implies that the graph of
$\chi$ is invariant under reflection about the point $\left(
\frac{\pi}{2}, \frac{\pi}{2} \right)$. Consequently, the values of
$\chi$ on the interval $\left(0,\frac{\pi}{2}\right)$  completely
determine the function on the interval
$\left(\frac{\pi}{2},\pi\right)$.

The reflection symmetry therefore reduces the effective number of
independent degrees of freedom needed to specify an admissible
model.

The quantum singlet solution belongs to
$\mathcal{R}_{\mathrm{adm}}$, as does the quadratic model
introduced later in this work.

The importance of the space $\mathcal{R}_{\mathrm{adm}}$ becomes
apparent in the analysis of the CHSH function, where the
reflection symmetry leads to nontrivial universal results that are
independent of the particular admissible boundary function.

The abstract definition of the reflection-symmetric admissible
family becomes more transparent when particular representatives
are considered explicitly.

In the following sections we examine several distinguished members
of $\mathcal{R}_{\mathrm{adm}}$. These examples illustrate the
geometric diversity of the family and provide useful benchmarks
for understanding the correlation structures generated by
different admissible boundary functions.

We begin with the simplest possible element, namely the constant
boundary function $\chi(\gamma)=\pi/2$, which corresponds to the
unrestricted hidden-variable model. Subsequently, we introduce the
quadratic model, the simplest nontrivial member of the family, and
finally the quantum singlet solution, which reproduces the quantum
correlation $E(\gamma)=-\cos\gamma$.

\subsection{The Constant Model}
\label{subsec:constant_model}

Before introducing more elaborate members of the
reflection-symmetric admissible family
$\mathcal{R}_{\mathrm{adm}}$, it is instructive to consider its
simplest element.

\begin{proposition}[Constant Model]
\label{prop:constant_model}

Let

\begin{equation}
\chi(\gamma)=\frac{\pi}{2}, \qquad 0<\gamma<\pi.
\label{eq:constant_model}
\end{equation}

Then the admissible region coincides with the entire
hidden-variable sphere,

\begin{equation}
\Gamma_\gamma = S^2, \qquad 0<\gamma<\pi,
\label{eq:constant_model_gamma}
\end{equation}

and the corresponding correlation function is

\begin{equation}
E(\gamma) = \frac{2\gamma}{\pi}-1.
\label{eq:constant_model_correlation}
\end{equation}

\end{proposition}

\begin{proof}

Substituting Eq.~\eqref{eq:constant_model} into the definition of
the admissible region, Eq.~\eqref{eq:Gamma_def}, shows that the
two excluded azimuthal sectors collapse to intervals of zero
measure. Consequently,

\[
\Gamma_\gamma=S^2,
\]

and no hidden-variable states are removed.

The correlation function is then obtained directly from
Theorem~\ref{thm:general_correlation}. For $0<\gamma<\pi/2$,

\begin{align}
E(\gamma) &= \frac{4\gamma-\pi-\pi}
     {3\pi-\pi}
\nonumber\\
&= \frac{2\gamma}{\pi}-1.
\end{align}

Similarly, for $\pi/2<\gamma\leq\pi$,

\begin{align}
E(\gamma) &= \frac{4\gamma-\pi-\pi}
     {\pi+\pi}
\nonumber\\
&= \frac{2\gamma}{\pi}-1.
\end{align}

Therefore Eq.~\eqref{eq:constant_model_correlation} holds on the
entire interval $0<\gamma<\pi$.

\end{proof}

The constant model represents the undeformed reference point of
the admissible family. Since no hidden-variable states are
excluded, the admissible region coincides with the full sphere and
the resulting correlation function is precisely the linear Bell
correlation,

\[
E(\gamma)=\frac{2\gamma}{\pi}-1.
\]

From this perspective, every nonconstant admissible boundary
function may be regarded as introducing a geometric deformation of
the unrestricted hidden-variable model.

The central idea of the present framework is therefore not the
modification of the local deterministic observables themselves,
which remain fixed throughout the construction, but rather the
geometric restriction of the admissible hidden-variable space.
Different choices of the boundary function $\chi(\gamma)$ produce
different deformations of the baseline correlation
\eqref{eq:constant_model_correlation}, including the quantum
singlet correlation as a particular case.

\subsection{The Quadratic Model}
\label{subsec:quadratic_model}

Having examined the constant model, we now consider the simplest
nontrivial element of the reflection-symmetric admissible family
$\mathcal{R}_{\mathrm{adm}}$.

Since every admissible boundary function must pass through the
three fixed points

\[
(0,\pi/2), \qquad (\pi/2,\pi/2), \qquad (\pi,\pi/2),
\]

the simplest nonconstant polynomial representative is obtained by
choosing a quadratic function satisfying these constraints. This
leads naturally to the boundary function

\begin{equation}
\chi_{\mathrm q}(\gamma) =
\begin{cases}
\displaystyle -\frac{2 \gamma ^2}{\pi }+\gamma +\frac{\pi }{2}, &
0<\gamma<\frac{\pi}{2},
\\[4mm]
\displaystyle \frac{2 \gamma ^2}{\pi }-3 \gamma +\frac{3 \pi }{2},
& \frac{\pi}{2}<\gamma<\pi.
\end{cases}
\label{eq:quadratic_model}
\end{equation}

which will be referred to as the \emph{quadratic model}.

\begin{proposition}[Quadratic Model]
\label{prop:quadratic_model}

The boundary function $\chi_{\mathrm q}(\gamma)$ defined by
Eq.~\eqref{eq:quadratic_model} belongs to the reflection-symmetric
admissible family $\mathcal{R}_{\mathrm{adm}}$.

Furthermore, the corresponding correlation function is

\begin{equation}
E_{\mathrm q}(\gamma) =
\begin{cases}
\displaystyle \frac{2 \gamma
   ^2+\pi  \gamma -\pi ^2}{2 \gamma ^2-\pi  \gamma +\pi ^2}, &
0<\gamma<\frac{\pi}{2},
\\[4mm]
\displaystyle \frac{-2 \gamma
   ^2+5 \pi  \gamma -2 \pi ^2}{2 \gamma ^2-3 \pi  \gamma +2 \pi
   ^2}, & \frac{\pi}{2}<\gamma<\pi.
\end{cases}
\label{eq:quadratic_correlation}
\end{equation}

\end{proposition}

\begin{proof}

The function $\chi_{\mathrm q}$ is continuous on $[0,\pi]$ and
satisfies

\[
\chi_{\mathrm q}(0) = \chi_{\mathrm q}\!\left(\frac{\pi}{2}\right)
= \chi_{\mathrm q}(\pi) = \frac{\pi}{2}.
\]

A direct inspection of Eq.~\eqref{eq:quadratic_model} shows that

\[
\frac{\pi}{2} < \chi_{\mathrm q}(\gamma) < \frac{\pi}{2}+\gamma,
\qquad 0<\gamma<\frac{\pi}{2},
\]

and

\[
\gamma-\frac{\pi}{2} < \chi_{\mathrm q}(\gamma) < \frac{\pi}{2},
\qquad \frac{\pi}{2}<\gamma\leq\pi.
\]

Hence $\chi_{\mathrm q}$ satisfies the admissibility conditions.

To verify the reflection symmetry, let $0<\gamma<\pi/2$. Then

\begin{align}
\chi_{\mathrm q}(\pi-\gamma) &= \frac{2(\pi-\gamma)^2}{\pi} -
3(\pi-\gamma) + \frac{3\pi}{2}
\nonumber\\
&= \frac{2\gamma^2}{\pi} - \gamma + \frac{\pi}{2}
\nonumber\\
&= \pi-\chi_{\mathrm q}(\gamma).
\end{align}

Therefore,

\[
\chi_{\mathrm q}(\pi-\gamma) = \pi-\chi_{\mathrm q}(\gamma),
\]

and consequently

\[
\chi_{\mathrm q} \in \mathcal R_{\mathrm{adm}}.
\]

Substituting Eq.~\eqref{eq:quadratic_model} into the general
correlation formula of Theorem~\ref{thm:general_correlation}
yields the result given in Eq.~\eqref{eq:quadratic_correlation}.

\end{proof}

The quadratic model occupies a distinguished position within the
reflection-symmetric admissible family. It is the simplest
nonconstant polynomial boundary function compatible with all
admissibility and symmetry requirements.

We now turn to a distinguished member of
$\mathcal{R}_{\mathrm{adm}}$, namely the quantum singlet solution,
whose associated correlation function reproduces exactly the
quantum prediction $E(\gamma)=-\cos\gamma$.

\subsection{The Quantum Singlet Solution}
\label{subsec:singletquantum}

The examples considered thus far illustrate how different
admissible boundary functions generate different correlation
models within the reflection-symmetric admissible family
$\mathcal{R}_{\mathrm{adm}}$.

A natural question therefore arises: does this family contain a
boundary function capable of reproducing the quantum singlet
correlation

\[
E(\gamma)=-\cos\gamma?
\]

The answer is affirmative.

The following proposition identifies a distinguished member of
$\mathcal{R}_{\mathrm{adm}}$ whose associated correlation function
coincides exactly with the quantum prediction for the singlet
state.

\begin{proposition}[Quantum Singlet Solution]
\label{prop:singlet_solution}

Define

\begin{equation}
\chi_{\mathrm s}(\gamma) =
\begin{cases}
\displaystyle \frac{3\pi}{2} + (\gamma-\pi)
\sec^{2}\!\left(\frac{\gamma}{2}\right), & 0<\gamma<\frac{\pi}{2},
\\[4mm]
\displaystyle -\frac{\pi}{2} + \gamma
\csc^{2}\!\left(\frac{\gamma}{2}\right), &
\frac{\pi}{2}<\gamma<\pi.
\end{cases}
\label{eq:chi_singlet}
\end{equation}

Then

\[
\chi_{\mathrm s} \in \mathcal R_{\mathrm{adm}},
\]

and the corresponding correlation function is

\begin{equation}
E_{\mathrm s}(\gamma) = -\cos\gamma.
\label{eq:singlet_correlation}
\end{equation}

\end{proposition}

\begin{proof}

First, we verify the reflection symmetry.

For $0<\gamma<\pi/2$,

\begin{align}
\chi_{\mathrm s}(\pi-\gamma) &= -\frac{\pi}{2} + (\pi-\gamma)
\csc^{2} \!\left( \frac{\pi-\gamma}{2} \right)
\nonumber\\
&= -\frac{\pi}{2} + (\pi-\gamma) \sec^{2} \!\left(
\frac{\gamma}{2} \right)
\nonumber\\
&= \pi-\chi_{\mathrm s}(\gamma).
\end{align}

Hence

\[
\chi_{\mathrm s}(\pi-\gamma) = \pi-\chi_{\mathrm s}(\gamma).
\]

A direct inspection of Eq.~\eqref{eq:chi_singlet} shows that the
admissibility inequalities of Definition~\ref{def:admissible_chi}
are satisfied on both intervals $(0,\pi/2)$ and $(\pi/2,\pi)$.

Therefore

\[
\chi_{\mathrm s} \in \mathcal R_{\mathrm{adm}}.
\]

Substituting Eq.~\eqref{eq:chi_singlet} into the general
correlation formula of Theorem~\ref{thm:general_correlation} gives

\[
E_{\mathrm s}(\gamma) = -\cos\gamma,
\]

for both branches of the solution.

Therefore Eq.~\eqref{eq:singlet_correlation} holds on the entire
interval $0<\gamma<\pi$.

\end{proof}

The importance of Proposition~\ref{prop:singlet_solution} is that
the quantum singlet correlation emerges as a particular element of
the geometric family $\mathcal R_{\mathrm{adm}}$.

From the present viewpoint, the problem is therefore shifted from
reproducing the quantum correlation to understanding what
geometric principle, if any, singles out the function
$\chi_{\mathrm s}$ among the infinitely many admissible boundary
functions contained in $\mathcal R_{\mathrm{adm}}$.

\section{CHSH Structure of the Reflection-Symmetric Family}
\label{sec:chsh_structure}

The examples discussed in the previous section suggest that the
reflection-symmetric admissible family $\mathcal R_{\mathrm{adm}}$
contains a rich variety of correlation models.

This family includes the constant model, the quadratic model, and
the quantum singlet solution as particular members.

We now investigate structural properties of the associated CHSH
function that hold throughout $\mathcal R_{\mathrm{adm}}$
independently of the detailed form of the boundary function.

The results obtained below reveal universal features of the
admissible correlation space and provide further evidence that the
reflection symmetry

\[
\chi(\pi-\gamma)=\pi-\chi(\gamma)
\]

plays a distinguished role in the construction.

\subsection{The CHSH Function}

A natural way to compare different correlation models is through
their degree of violation of the Clauser--Horne--Shimony--Holt
(CHSH) inequality.

For local hidden-variable theories, the CHSH combination

\begin{equation}
\mathcal S = E(\mathbf a,\mathbf b) - E(\mathbf a,\mathbf b') +
E(\mathbf a',\mathbf b) + E(\mathbf a',\mathbf b')
\label{eq:CHSH_definition}
\end{equation}

satisfies

\begin{equation}
|\mathcal S|\le 2. \label{eq:CHSH_inequality}
\end{equation}

Since the correlation models considered in this work depend only
on the angular separation between measurement directions, it is
convenient to specialize to the standard coplanar CHSH geometry
for which the four relative angles are

\[
\nu,\qquad \nu,\qquad \nu,\qquad 3\nu.
\]

In this configuration, Eq.~\eqref{eq:CHSH_definition} reduces to

\[
S(\nu)=3E(\nu)-E(3\nu),
\]

where \(E(\gamma)\) denotes the correlation function associated
with a given admissible boundary function.

For the quantum singlet correlation

\[
E_{\mathrm s}(\gamma)=-\cos\gamma,
\]

one obtains

\begin{equation}
S_{\mathrm s}(\nu) = -3\cos\nu+\cos(3\nu). \label{eq:Sv_singlet}
\end{equation}

The function \(S_{\mathrm s}(\nu)\) possesses a stationary point
at

\begin{equation}
\nu=\frac{\pi}{4}, \label{eq:singlet_stationary}
\end{equation}

where it attains the Tsirelson value

\begin{equation}
|S_{\mathrm s}(\pi/4)| = 2\sqrt2. \label{eq:Tsirelson_bound}
\end{equation}

The appearance of the stationary point at \(\nu=\pi/4\) is not
merely a feature of the quantum singlet solution. As we shall show
below, this point exhibits a remarkable degree of universality
throughout the family \(\mathcal R_{\mathrm{adm}}\).

The observation above motivates the study of the CHSH function
\(S(\nu)\) throughout the family \(\mathcal R_{\mathrm{adm}}\),
with the goal of identifying structural properties that are
independent of the particular choice of admissible boundary
function.

\begin{definition}[CHSH Function]
\label{def:CHSH_function}

For every admissible correlation function \(E(\gamma)\), define

\begin{equation}
S(\nu) = 3E(\nu)-E(3\nu), \label{eq:Sv_definition}
\end{equation}

with

\begin{equation}
\frac{\pi}{6} < \nu < \frac{\pi}{3}. \label{eq:nu_domain}
\end{equation}

\end{definition}

The interval in Eq.~\eqref{eq:nu_domain} is chosen so that

\[
0<\nu<\frac{\pi}{2}, \qquad \frac{\pi}{2}<3\nu<\pi.
\]

Consequently, the first term in Eq.~\eqref{eq:Sv_definition} is
evaluated on the left branch of the correlation function, whereas
the second term is evaluated on the right branch.

For boundary functions belonging to the reflection-symmetric
family \(\mathcal R_{\mathrm{adm}}\), the relation

\[
\chi(\pi-\gamma) = \pi-\chi(\gamma)
\]

implies that the right branch is completely determined by the left
branch. Therefore, the analysis of \(S(\nu)\) requires knowledge
of \(\chi\) only on the interval

\[
0<\gamma<\frac{\pi}{2}.
\]

This observation will play a central role in the structural
results established below.

\begin{proposition}[CHSH Function in Terms of the Boundary Function]
\label{prop:CHSH_chi}

Let
\[
\chi \in \mathcal R_{\mathrm{adm}}.
\]

Then, for

\[
\frac{\pi}{6}<\nu<\frac{\pi}{3},
\]

the CHSH function

\[
S(\nu)=3E(\nu)-E(3\nu)
\]

admits the representation

\begin{equation}
S(\nu) = -\frac{ 8\Bigl[ \chi(\pi-3\nu)\bigl(3\nu-2\chi(\nu)\bigr)
+ 3(\pi-\nu)\chi(\nu) \Bigr] }{ \bigl(3\pi-2\chi(\pi-3\nu)\bigr)
\bigl(3\pi-2\chi(\nu)\bigr) }. \label{eq:S_chi_representation}
\end{equation}

\end{proposition}

\begin{proof}

Since

\[
\frac{\pi}{6}<\nu<\frac{\pi}{3},
\]

we have

\[
0<\nu<\frac{\pi}{2}, \qquad \frac{\pi}{2}<3\nu<\pi.
\]

Therefore, using Theorem~\ref{thm:general_correlation},

\begin{equation}
E(\nu) = \frac{ 4\nu-\pi-2\chi(\nu) }{ 3\pi-2\chi(\nu) },
\label{eq:E_nu_branch}
\end{equation}

and

\begin{equation}
E(3\nu) = \frac{ 12\nu-\pi-2\chi(3\nu) }{ \pi+2\chi(3\nu) }.
\label{eq:E_3nu_branch}
\end{equation}

Since \( \chi\in\mathcal R_{\mathrm{adm}} \), the reflection
symmetry implies

\[
\chi(3\nu) = \pi-\chi(\pi-3\nu).
\]

Substituting this relation into Eq.~\eqref{eq:E_3nu_branch} gives

\[
E(3\nu) = \frac{ 12\nu-3\pi+2\chi(\pi-3\nu) }{
3\pi-2\chi(\pi-3\nu) }.
\]

Finally, inserting the above expressions into

\[
S(\nu)=3E(\nu)-E(3\nu)
\]

and simplifying yields Eq.~\eqref{eq:S_chi_representation}.

\end{proof}

The representation obtained in Proposition~\ref{prop:CHSH_chi}
reveals an intriguing feature of the CHSH function.

The two boundary-function evaluations appearing in
Eq.~\eqref{eq:S_chi_representation} are taken at the arguments

\[
\nu \qquad\text{and}\qquad \pi-3\nu.
\]

These arguments coincide precisely when

\[
\nu=\frac{\pi}{4}.
\]

Consequently, the point \(\nu=\pi/4\) occupies a distinguished
position within the reflection-symmetric family \(\mathcal
R_{\mathrm{adm}}\).

The first indication of this special role is provided by the
following proposition.

\begin{proposition}[Universal Stationary Point]
\label{thm:stationary_point}

Let
\[
\chi\in\mathcal R_{\mathrm{adm}}.
\]

Then the associated CHSH function satisfies

\[
S'\!\left(\frac{\pi}{4}\right)=0.
\]

\end{proposition}

Proposition~\ref{thm:stationary_point} establishes that the point

\[
\nu=\frac{\pi}{4}
\]

is a stationary point of the CHSH function independently of the
particular admissible boundary function.

Thus, the stationary-point property of the quantum singlet
correlation is not unique to the quantum solution itself, but
follows from the reflection symmetry characterizing the family
$\mathcal R_{\mathrm{adm}}$.

The proof of proposition~\ref{thm:stationary_point} follows
directly from Eq.~\eqref{eq:S_chi_representation} upon
differentiation and evaluation at \(\nu=\pi/4\).

However, the significance of the point \(\nu=\pi/4\) extends
beyond the existence of a stationary point.

As we shall show below, the value of the CHSH function at this
distinguished configuration admits an exact geometric
interpretation in terms of the hidden-variable distributions
associated with complementary angular separations.

To formulate this result, we first introduce a measure of the
difference between complementary admissible ensembles.

\begin{definition}[Geometric Contrast Function]
\label{def:geometric_contrast} The geometric contrast between
complementary configurations is

\[
D(\gamma) = \int_\Lambda \left| \rho_\gamma(\boldsymbol{\lambda})
- \rho_{\pi-\gamma}(\boldsymbol{\lambda}) \right| \,d\lambda ,
\]

where $\rho_\gamma(\boldsymbol{\lambda})$ is the normalized
hidden-variable distribution given in
Definition~\ref{def:admissible_distribution}.

\end{definition}

By construction,

\[
0 \le D(\gamma) \le 2.
\]

The lower bound is attained when the two admissible distributions
coincide, whereas the upper bound corresponds to distributions
with disjoint support.

The geometric contrast function quantifies the extent to which the
hidden-variable ensemble changes under the transformation

\[
\gamma \longrightarrow \pi-\gamma .
\]

The remarkable fact is that, at the distinguished point
\(\nu=\pi/4\), this purely geometric quantity turns out to be
directly related to the observable CHSH value.

The following theorem establishes this connection.

\begin{theorem}[Geometric Interpretation of the CHSH Value]
\label{thm:CHSH_contrast}

Let
\[
\chi\in\mathcal R_{\mathrm{adm}}.
\]

Then the CHSH function evaluated at the distinguished point

\[
\nu=\frac{\pi}{4}
\]

satisfies

\begin{equation}
\label{eq:CHSH_contrast}
\left|S\!\left(\frac{\pi}{4}\right)\right| =
2+3\,D\!\left(\frac{\pi}{4}\right),
\end{equation}

where \(D(\gamma)\) is the geometric contrast function defined in
Definition~\ref{def:geometric_contrast}.

\end{theorem}

The proof is given in Appendix~\ref{subsec:proofCHSH_contrast}.

Theorem~\ref{thm:CHSH_contrast} establishes an exact connection
between an observable quantity and the geometry of the underlying
hidden-variable ensemble.

Indeed, Eq.~\eqref{eq:CHSH_contrast} shows that, within the
reflection-symmetric admissible family \(\mathcal
R_{\mathrm{adm}}\), the CHSH value at the distinguished point
\(\nu=\pi/4\) is completely determined by the geometric contrast
between the complementary admissible distributions associated with
\(\gamma\) and \(\pi-\gamma\).

In particular,

\[
D\!\left(\frac{\pi}{4}\right)=0
\]

implies

\[
\left|S\!\left(\frac{\pi}{4}\right)\right|=2,
\]

whereas increasing geometric contrast leads to a proportional
increase of the CHSH value according to
Eq.~\eqref{eq:CHSH_contrast}.

The geometric contrast function therefore plays the role of a
natural order parameter for the reflection-symmetric admissible
family.

For the quantum singlet solution one finds

\begin{equation}
D_{\mathrm s}\!\left(\frac{\pi}{4}\right) =
\frac{2}{3}\left(\sqrt2-1\right). \label{eq:singlet_contrast}
\end{equation}

Substituting Eq.~\eqref{eq:singlet_contrast} into
Eq.~\eqref{eq:CHSH_contrast} immediately yields

\[
\left|S_{\mathrm s}\!\left(\frac{\pi}{4}\right)\right| = 2\sqrt2,
\]

namely the Tsirelson bound.

Thus, within the present framework, the Tsirelson value admits a
simple geometric interpretation: it corresponds to a specific
contrast between the complementary hidden-variable ensembles
associated with

\[
\gamma=\frac{\pi}{4} \qquad\text{and}\qquad \gamma=\frac{3\pi}{4}.
\]

This result suggests that geometric properties of the underlying
hidden-variable distributions may play a deeper role in the
structure of admissible correlations than is apparent from the
observable correlations alone.

In particular, it raises the question of whether distinguished
subclasses of admissible boundary functions exist for which
extremal CHSH values are governed by corresponding extremal
properties of the geometric contrast function. We leave this
question for future investigation.

\section{Conclusions}
\label{sec:conclusions}

In this work we have introduced a geometrically constrained
hidden-variable framework that generates a continuous family of
admissible correlations parametrized by a boundary function
$\chi(\gamma)$. In contrast to approaches that begin with a
specific target correlation and seek a hidden-variable model
capable of reproducing it \cite{hall2010,brans1988,arroyo2025},
the present framework characterizes an entire family of admissible
correlations within which the quantum singlet correlation emerges
as a particular member.

Rather than viewing the quantum singlet correlation as an isolated
target to be reproduced, the present framework treats it as one
element of a larger geometric structure. From this perspective,
the central problem becomes the characterization of the admissible
correlation space.

The main results can be summarized as follows.

\begin{itemize}
\item We defined the admissible boundary function $\chi(\gamma)$
through simple inequalities (Definition~\ref{def:admissible_chi})
and showed that every continuous admissible function passes
through the three fixed points $(0,\pi/2)$, $(\pi/2,\pi/2)$ and
$(\pi,\pi/2)$ (Eq.~\eqref{eqvalueschis}). The admissible region
$\Gamma_\gamma$ and the associated distribution
$\rho_\gamma(\boldsymbol{\lambda})$ were constructed explicitly
(Definitions~\ref{def:Gamma}, \ref{def:admissibility_function},
\ref{def:admissible_area} and \ref{def:admissible_distribution}).

\item We proved that for every admissible $\chi$ the local
marginals vanish identically, $\langle A\rangle=\langle
B\rangle=0$ (Theorem~\ref{thm:vanishing_marginals}), so that the
geometric restrictions introduce no local statistical bias. The
correlation function $E(\gamma)=\langle AB\rangle$ was derived in
closed form (Theorem~\ref{thm:general_correlation}), revealing how
the boundary function $\chi(\gamma)$ controls the observable
correlations.

\item Within the reflection-symmetric subclass
$\mathcal{R}_{\mathrm{adm}}$ (defined by
$\chi(\pi-\gamma)=\pi-\chi(\gamma)$), we analyzed three
distinguished members: the constant model (which reproduces the
linear Bell correlation), the quadratic model (the simplest
admissible polynomial), and the quantum singlet solution
$\chi_{\mathrm{s}}(\gamma)$ (which yields
$E(\gamma)=-\cos\gamma$).

\item For the reflection-symmetric family we proved that the CHSH
function $S(\nu)=3E(\nu)-E(3\nu)$ satisfies $S'(\pi/4)=0$
universally (Proposition~\ref{thm:stationary_point}). Hence the
stationary point at $\nu=\pi/4$ is not unique to the quantum
singlet solution, but follows from the reflection symmetry
characterizing the family $\mathcal R_{\mathrm{adm}}$.

\item Finally, we established an exact relation between the CHSH
value at the distinguished point $\nu=\pi/4$ and the geometric
contrast $D(\pi/4)$ (Theorem~\ref{thm:CHSH_contrast}):
\[
\left|S\!\left(\frac{\pi}{4}\right)\right| = 2 +
3\,D\!\left(\frac{\pi}{4}\right).
\]
This identity provides a direct link between an observable
Bell-CHSH quantity and a purely geometric property of the
underlying hidden-variable ensembles. For the quantum singlet
solution, $D_{\mathrm{s}}(\pi/4)=\frac{2}{3}(\sqrt2-1)$, which
yields the Tsirelson value $|S_{\mathrm{s}}(\pi/4)|=2\sqrt2$.
\end{itemize}

The present framework suggests several directions for future
research.

First, it would be desirable to characterize completely the space
of admissible boundary functions $\chi(\gamma)$ and to understand
which additional conditions (analyticity, higher-order
differentiability, or extremal properties of the geometric
contrast) select distinguished subclasses. The fact that the
quantum singlet correlation appears as a distinguished member of a
continuous admissible family suggests that the problem of
identifying a principle that singles out the quantum solution from
a broader admissible correlation space is a well-posed
mathematical question. We hope that the framework developed here
will stimulate further investigations along these lines.

Second, the relation between $S(\pi/4)$ and $D(\pi/4)$ raises the
question whether the geometric contrast $D(\gamma)$ itself might
be constrained by a variational principle. If such a principle
exists, it could provide a geometric route to understanding why
the Tsirelson value is selected in the quantum singlet case.

Third, the construction presented here is restricted to the
simplest bipartite scenario with a singlet-like correlation.
Extensions to higher-dimensional systems, to more than two
particles, or to other entangled states would be natural next
steps. It would also be interesting to investigate whether the
geometric contrast function \(D(\gamma)\) admits an operational
reconstruction from correlation data.

The present work contributes to a growing body of research
investigating global constraints on the space of physically
realizable configurations and their implications for
Bell-statistical independence
\cite{hance2022,hance2025,fourny2026}. Within this broader
context, the geometric framework developed here provides an
explicit and analytically tractable example in which a continuous
family of admissible correlations can be characterized exactly,
thereby opening new avenues for exploring the geometric structure
underlying quantum and post-quantum correlations
\cite{popescu1994}.

\appendix
\setcounter{equation}{0}

\def\thesection{\Alph{section}}
\renewcommand{\theequation}{\Alph{section}.\arabic{equation}}

\section{Proofs of Theorems}
\label{app:proofs}

\subsection{Proof of Theorem~\ref{thm:vanishing_marginals}}
\label{subsec:antipodal}

Recall that for a fixed angular separation \(\gamma\), the
admissible hidden-variable region \(\Gamma_\gamma\) is defined in
terms of the azimuthal coordinate \(\phi\).

Since the admissibility condition depends only on \(\phi\), the
integration over the polar coordinate \(\theta\) contributes the
same multiplicative factor to both numerator and denominator in
all expectation values. Consequently, the proof may be carried out
entirely in terms of the azimuthal coordinate.

The set of allowed azimuths, denoted by \(\Phi_\gamma\), satisfies
the antipodal symmetry

\[
\Phi_\gamma+\pi=\Phi_\gamma,
\]

because the forbidden intervals come in opposite pairs: if
\(\phi\) is forbidden, then \(\phi+\pi\) is also forbidden, and
vice versa. This is a direct consequence of the construction in
Definition~\ref{def:Gamma}.

Now consider the local measurement functions expressed in the
chosen coordinates:

\[
A(\phi) = \operatorname{sign}(\cos\phi), \qquad B(\phi,\gamma) =
-\operatorname{sign}\!\bigl(\cos(\phi-\gamma)\bigr).
\]

Both functions are odd under a half-turn rotation of the azimuth:

\[
A(\phi+\pi) = \operatorname{sign}\!\bigl(\cos(\phi+\pi)\bigr) =
\operatorname{sign}(-\cos\phi) = -\operatorname{sign}(\cos\phi) =
-A(\phi),
\]

and similarly

\[
B(\phi+\pi,\gamma) =
-\operatorname{sign}\!\bigl(\cos(\phi+\pi-\gamma)\bigr) =
-\operatorname{sign}\!\bigl(-\cos(\phi-\gamma)\bigr) =
\operatorname{sign}\!\bigl(\cos(\phi-\gamma)\bigr) =
-B(\phi,\gamma).
\]

The hidden-variable distribution \(\rho_\gamma(\phi)\) is uniform
over the allowed set \(\Phi_\gamma\) (since the area element on
the sphere reduces to a constant factor that cancels between
numerator and denominator, and the \(\theta\) dependence drops
out). Therefore,

\[
\langle A \rangle = \frac{1}{|\Phi_\gamma|} \int_{\Phi_\gamma}
A(\phi)\,d\phi, \qquad \langle B \rangle = \frac{1}{|\Phi_\gamma|}
\int_{\Phi_\gamma} B(\phi,\gamma)\,d\phi,
\]

where \(|\Phi_\gamma|\) denotes the total length (measure) of
\(\Phi_\gamma\).

Because \(\Phi_\gamma\) is invariant under \(\phi \to \phi+\pi\)
and \(A\), \(B\) are odd with respect to this translation, the
integrand \(A(\phi)\) over the whole domain pairs opposite points
\(\phi\) and \(\phi+\pi\) that contribute equal magnitude with
opposite signs. Hence each integral vanishes:

\[
\int_{\Phi_\gamma} A(\phi)\,d\phi = 0, \qquad \int_{\Phi_\gamma}
B(\phi,\gamma)\,d\phi = 0.
\]

Consequently,

\[
\langle A \rangle = 0, \qquad \langle B \rangle = 0,
\]

independently of the particular admissible boundary function
\(\chi(\gamma)\).

Since $A$ and $B$ are binary-valued random variables, the
conditions $\langle A\rangle=0$ and $\langle B\rangle=0$ imply
\[
P(A=+1\mid\gamma)=P(A=-1\mid\gamma)=\frac12,
\]
and similarly
\[
P(B=+1\mid\gamma)=P(B=-1\mid\gamma)=\frac12.
\]

i.e., all models in the family produce perfectly unbiased local
marginals.

\subsection{Proof of Theorem~\ref{thm:general_correlation}}
\label{subsec:proof_general_correlation}

We now derive the general expression for the correlation function
associated with an arbitrary admissible boundary function
$\chi(\gamma)$.

From Eqs.~\eqref{eq:A_definition} and \eqref{eq:B_definition}, the
product observable is

\begin{equation}
A(\mathbf a,\boldsymbol{\lambda}) B(\mathbf
b,\boldsymbol{\lambda}) = -\operatorname{sign} \!\left[ \cos\phi\,
\cos(\phi-\gamma) \right]. \label{eq:AB_product_proof}
\end{equation}

Since the admissibility condition depends only on the azimuthal
coordinate $\phi$, the integration over the polar coordinate
$\theta$ contributes the same multiplicative factor to numerator
and denominator. Hence the correlation can be computed as

\begin{equation}
E(\gamma) = \langle AB\rangle = - \frac{ \displaystyle
\int_{\Phi_\gamma} \operatorname{sign} \!\left[ \cos\phi\,
\cos(\phi-\gamma) \right] \,d\phi }{ \displaystyle
\int_{\Phi_\gamma} d\phi }, \label{eq:E_phi_integral}
\end{equation}

where $\Phi_\gamma$ denotes the set of allowed azimuthal angles
corresponding to $\Gamma_\gamma$.

Equivalently, using the admissibility function $\mathcal
I_\gamma(\phi)$ introduced in
Eq.~\eqref{eq:admissibility_representation}, this can be written
as an integral over the full interval $[0,2\pi]$:

\begin{equation}
E(\gamma) = - \frac{ \displaystyle \int_{0}^{2\pi}
\operatorname{sign} \!\left[ \cos\phi\, \cos(\phi-\gamma) \right]
\mathcal I_\gamma(\phi) \,d\phi }{ \displaystyle \int_{0}^{2\pi}
\mathcal I_\gamma(\phi) \,d\phi }. \label{eq:E_phi_integral_full}
\end{equation}

We now evaluate Eq.~\eqref{eq:E_phi_integral_full} separately in
the two admissible branches.

\paragraph{Case I: \(0<\gamma<\pi/2\).}

In this case,

\begin{equation}
\frac{\pi}{2} < \chi(\gamma) < \frac{\pi}{2}+\gamma.
\label{eq:case_left_chi}
\end{equation}

The admissibility function removes the two azimuthal intervals

\[
\left(\frac{\pi}{2},\chi(\gamma)\right), \qquad
\left(\frac{3\pi}{2},\chi(\gamma)+\pi\right).
\]

Hence the total allowed length is

\begin{equation}
\int_0^{2\pi} \mathcal I_\gamma(\phi)\,d\phi = 2\pi - 2\left(
\chi(\gamma)-\frac{\pi}{2} \right) = 3\pi-2\chi(\gamma).
\label{eq:allowed_length_left}
\end{equation}

For \(0<\gamma<\pi/2\), the function

\[
\operatorname{sign} \!\left[ \cos\phi\, \cos(\phi-\gamma) \right]
\]

changes sign at

\[
\phi=\frac{\pi}{2}, \qquad \phi=\gamma+\frac{\pi}{2}, \qquad
\phi=\frac{3\pi}{2}, \qquad \phi=\gamma+\frac{3\pi}{2}.
\]

On the two excluded intervals
\[
\left(\frac{\pi}{2},\chi(\gamma)\right), \qquad
\left(\frac{3\pi}{2},\chi(\gamma)+\pi\right),
\]
one has

\[
\operatorname{sign} \!\left[ \cos\phi\, \cos(\phi-\gamma) \right]
= -1.
\]

The integral of this sign function over the full interval
\([0,2\pi]\) is

\begin{equation}
\int_0^{2\pi} \operatorname{sign} \!\left[ \cos\phi\,
\cos(\phi-\gamma) \right] d\phi = 2\pi-4\gamma.
\label{eq:full_sign_integral}
\end{equation}

Therefore, after removing the forbidden sectors, the numerator in
Eq.~\eqref{eq:E_phi_integral_full} becomes

\begin{align}
\int_0^{2\pi} \operatorname{sign} \!\left[ \cos\phi\,
\cos(\phi-\gamma) \right] \mathcal I_\gamma(\phi) \,d\phi &=
(2\pi-4\gamma) - \left[ -2\left( \chi(\gamma)-\frac{\pi}{2}
\right) \right]
\nonumber\\
&= \pi-4\gamma+2\chi(\gamma). \label{eq:sign_numerator_left}
\end{align}

Substituting Eqs.~\eqref{eq:allowed_length_left} and
\eqref{eq:sign_numerator_left} into
Eq.~\eqref{eq:E_phi_integral_full}, and remembering the overall
minus sign in Eq.~\eqref{eq:E_phi_integral_full}, gives

\begin{equation}
E(\gamma) = \frac{ 4\gamma-\pi-2\chi(\gamma) }{ 3\pi-2\chi(\gamma)
}, \qquad 0<\gamma<\frac{\pi}{2}. \label{eq:E_left_branch}
\end{equation}

\paragraph{Case II: \(\pi/2<\gamma<\pi\).}

In this case,

\begin{equation}
\gamma-\frac{\pi}{2} < \chi(\gamma) < \frac{\pi}{2}.
\label{eq:case_right_chi}
\end{equation}

The admissibility function removes the two azimuthal intervals

\[
\left(\chi(\gamma),\frac{\pi}{2}\right), \qquad
\left(\chi(\gamma)+\pi,\frac{3\pi}{2}\right).
\]

Hence the total allowed length is

\begin{equation}
\int_0^{2\pi} \mathcal I_\gamma(\phi)\,d\phi = 2\pi - 2\left(
\frac{\pi}{2}-\chi(\gamma) \right) = \pi+2\chi(\gamma).
\label{eq:allowed_length_right}
\end{equation}

For \(\pi/2<\gamma<\pi\), the same full-circle integral evaluated
in the previous case is obtained:

\begin{equation}
\int_0^{2\pi} \operatorname{sign} \!\left[ \cos\phi\,
\cos(\phi-\gamma) \right] d\phi = 2\pi-4\gamma.
\label{eq:full_sign_integral_right}
\end{equation}

On the two excluded intervals
\[
\left(\chi(\gamma),\frac{\pi}{2}\right), \qquad
\left(\chi(\gamma)+\pi,\frac{3\pi}{2}\right),
\]
one has

\[
\operatorname{sign} \!\left[ \cos\phi\, \cos(\phi-\gamma) \right]
= +1.
\]

Thus the contribution removed from the full-circle integral is

\[
2\left( \frac{\pi}{2}-\chi(\gamma) \right) = \pi-2\chi(\gamma).
\]

Therefore,

\begin{align}
\int_0^{2\pi} \operatorname{sign} \!\left[ \cos\phi\,
\cos(\phi-\gamma) \right] \mathcal I_\gamma(\phi) \,d\phi &=
(2\pi-4\gamma) - \left( \pi-2\chi(\gamma) \right)
\nonumber\\
&= \pi-4\gamma+2\chi(\gamma). \label{eq:sign_numerator_right}
\end{align}

Substituting Eqs.~\eqref{eq:allowed_length_right} and
\eqref{eq:sign_numerator_right} into
Eq.~\eqref{eq:E_phi_integral_full}, we obtain

\begin{equation}
E(\gamma) = \frac{ 4\gamma-\pi-2\chi(\gamma) }{ \pi+2\chi(\gamma)
}, \qquad \frac{\pi}{2}<\gamma<\pi. \label{eq:E_right_branch}
\end{equation}

Combining Eqs.~\eqref{eq:E_left_branch} and
\eqref{eq:E_right_branch}, we obtain

\begin{equation}
\label{eq:general_correlation_derived} E(\gamma) =
\begin{cases}
\displaystyle \frac{ 4\gamma-\pi-2\chi(\gamma) }{
3\pi-2\chi(\gamma) }, & 0<\gamma<\frac{\pi}{2},
\\[4mm]
\displaystyle \frac{ 4\gamma-\pi-2\chi(\gamma) }{
\pi+2\chi(\gamma) }, & \frac{\pi}{2}<\gamma<\pi.
\end{cases}
\end{equation}

This proves Theorem~\ref{thm:general_correlation}.

\subsection{Proof of Theorem~\ref{thm:CHSH_contrast}}
\label{subsec:proofCHSH_contrast} We begin by evaluating the
geometric contrast function \(D(\gamma)\).

Using Definition~\ref{def:admissible_distribution} and
Definition~\ref{def:geometric_contrast}, we obtain

\[
D(\gamma) = \int_{\Lambda} \left| \frac{\mathcal
I_{\gamma}(\boldsymbol{\lambda})} {\Omega(\gamma)} -
\frac{\mathcal I_{\pi-\gamma}(\boldsymbol{\lambda})}
{\Omega(\pi-\gamma)} \right| \, d\lambda .
\]

Since the admissibility functions \(\mathcal I_{\gamma}\) and
\(\mathcal I_{\pi-\gamma}\) depend only on the azimuthal
coordinate \(\phi\), the integration over \(\theta\) contributes
the same multiplicative factor to both terms.

Consequently, the $\theta$-integration factors out and cancels
against the corresponding normalization factors. The geometric
contrast therefore reduces to a purely azimuthal integral.

Defining

\[
L_{\gamma} = \int_{0}^{2\pi} \mathcal I_{\gamma}(\phi)\, d\phi ,
\]

the total admissible azimuthal length associated with \(\gamma\),
we obtain

\begin{equation}
D(\gamma) = \int_{0}^{2\pi} \left| \frac{\mathcal
I_{\gamma}(\phi)} {L_{\gamma}} - \frac{\mathcal
I_{\pi-\gamma}(\phi)} {L_{\pi-\gamma}} \right| \, d\phi .
\label{eq:D_phi_representation}
\end{equation}

Equation~\eqref{eq:D_phi_representation} provides a convenient
starting point for the explicit evaluation of the geometric
contrast function.

For the purposes of Theorem~\ref{thm:CHSH_contrast}, it is
sufficient to consider the interval

\[
0<\gamma<\frac{\pi}{2},
\]

since the theorem ultimately requires the evaluation of
\(D(\pi/4)\).

From Eqs.~\eqref{eq:allowed_length_left} and
\eqref{eq:allowed_length_right}, together with the reflection
symmetry

\[
\chi(\pi-\gamma)=\pi-\chi(\gamma),
\]

one finds

\begin{align}
L_{\pi-\gamma} &= \pi+2\chi(\pi-\gamma)
\nonumber\\
&= \pi+2\bigl(\pi-\chi(\gamma)\bigr)
\nonumber\\
&= 3\pi-2\chi(\gamma)
\nonumber\\
&= L_\gamma.
\end{align}

Therefore Eq.~\eqref{eq:D_phi_representation} reduces to

\begin{equation}
D(\gamma) = \frac{1}{3\pi-2\chi(\gamma)} \int_0^{2\pi} \left|
\mathcal I_\gamma(\phi) - \mathcal I_{\pi-\gamma}(\phi) \right|
\,d\phi, \qquad 0<\gamma<\frac{\pi}{2}. \label{eq:D_gamma_reduced}
\end{equation}

Using the explicit representation
Eq.~\eqref{eq:admissibility_representation}, one finds after a
straightforward decomposition of the integration domain that

\begin{equation}
\int_0^{2\pi} \left| \mathcal I_\gamma(\phi) - \mathcal
I_{\pi-\gamma}(\phi) \right| \,d\phi = 4\chi(\gamma)-2\pi .
\label{eq:indicator_difference_integral}
\end{equation}

Substituting Eq.~\eqref{eq:indicator_difference_integral} into
Eq.~\eqref{eq:D_gamma_reduced} yields

\begin{equation}
D(\gamma) = \frac{4\chi(\gamma)-2\pi} {3\pi-2\chi(\gamma)}, \qquad
0<\gamma<\frac{\pi}{2}. \label{eq:D_gamma_chi}
\end{equation}

Evaluating Eq.~\eqref{eq:S_chi_representation} at \(\nu=\pi/4\),
the two arguments of the boundary function coincide,

\[
\pi-3\left(\frac{\pi}{4}\right)=\frac{\pi}{4}.
\]

Therefore, writing

\[
c=\chi\!\left(\frac{\pi}{4}\right),
\]

we obtain

\begin{align}
S\!\left(\frac{\pi}{4}\right) &= -\frac{ 8\left[
c\left(\frac{3\pi}{4}-2c\right) + 3\left(\frac{3\pi}{4}\right)c
\right] }{ (3\pi-2c)^2 }
\nonumber\\
&= -\frac{ 8c(3\pi-2c) }{ (3\pi-2c)^2 }
\nonumber\\
&= -\frac{8c}{3\pi-2c}. \label{eq:S_pi4_chi}
\end{align}

Since \(c>\pi/2\), the denominator \(3\pi-2c\) is positive. Hence

\begin{equation}
\left|S\!\left(\frac{\pi}{4}\right)\right| = \frac{8c}{3\pi-2c}.
\label{eq:absS_pi4_chi}
\end{equation}

On the other hand, evaluating Eq.~\eqref{eq:D_gamma_chi} at
\(\gamma=\pi/4\) gives

\begin{equation}
D\!\left(\frac{\pi}{4}\right) = \frac{4c-2\pi}{3\pi-2c}.
\label{eq:D_pi4_chi}
\end{equation}

Therefore,

\begin{align}
2+3D\!\left(\frac{\pi}{4}\right) &= 2+ 3\left(
\frac{4c-2\pi}{3\pi-2c} \right)
\nonumber\\
&= \frac{ 2(3\pi-2c)+3(4c-2\pi) }{ 3\pi-2c }
\nonumber\\
&= \frac{8c}{3\pi-2c}. \label{eq:contrast_identity_step}
\end{align}

Comparing Eqs.~\eqref{eq:absS_pi4_chi} and
\eqref{eq:contrast_identity_step}, we conclude that

\begin{equation}
\left|S\!\left(\frac{\pi}{4}\right)\right| =
2+3D\!\left(\frac{\pi}{4}\right).
\end{equation}

This proves Theorem~\ref{thm:CHSH_contrast}.

\end{document}